 \documentclass[11pt]{article}%
\usepackage{graphicx}
\usepackage[ ruled, boxed, linesnumbered]{algorithm2e}
\usepackage{setspace}
\usepackage{amssymb,amsfonts}
\usepackage[margin = 1in]{geometry}
\usepackage{microtype}

\usepackage{url}
\usepackage{color,microtype}
\usepackage{graphicx}
\usepackage{amssymb,amsfonts,amsmath,amsthm}
\usepackage{textcomp}
\usepackage{multirow}
\usepackage{array}
\usepackage{xspace}

\DeclareMathOperator{\polylog}{polylog}
\DeclareMathOperator{\poly}{poly}

\usepackage[sort]{cite}

\newcommand{\C}{{\mathcal{C}}}

\newcommand{\ind}{\textup{Index}\xspace}

\newcommand{\tO}{\tilde{O}}

\newcommand{\sat}{\textsf{SAT}\xspace}
\newcommand{\ksat}{\textup{$k$-SAT}\xspace}
\newcommand{\maxsat}{\textup{Max-SAT}\xspace}
\newcommand{\minsat}{\textup{Min-SAT}\xspace}
\newcommand{\minksat}{\textup{Min-$k$-SAT}\xspace}
\newcommand{\maxksat}{\textup{Max-$k$-SAT}\xspace}
\newcommand{\maxtwosat}{\textup{Max-2-SAT}\xspace}
\newcommand{\maxandsat}{\textup{Max-AND-SAT}\xspace}

\newcommand{\compl}[1]{\overline{#1}}

\DeclareMathOperator{\opt}{\textup{OPT}}

\DeclareMathOperator{\true}{\textup{true}}
\DeclareMathOperator{\false}{\textup{false}}

\renewcommand{\epsilon}{\varepsilon}

\renewcommand{\exp}[1]{{\textup{exp}}\left( #1 \right)}

\newcommand{\expec}[1]{ \mathbb{E}\left [ #1 \right ]}

\newtheorem{theorem}{Theorem}

\newtheorem{lemma}[theorem]{Lemma}

\newtheorem{claim}[theorem]{Claim}

\newcommand{\prob}[1]{ \mathbb{P} \left( #1 \right)}

\begin{document}
\title{Revisiting Maximum Satisfiability and Related Problems in Data Streams \footnote{An extended abstract of this paper will appear in the  28th International Computing and Combinatorics Conference  (COCOON 2022). }}
%
%
\author{Hoa T. Vu \footnote{San Diego State University, San Diego, CA, USA. Email: hvu2@sdsu.edu.}}
\date{}
%
\maketitle              
\begin{abstract}

We revisit the maximum satisfiability problem (\maxsat)  in the data stream model. In this problem, the stream consists of $m$ clauses that are disjunctions of literals drawn from $n$ Boolean variables. The objective is to find an  assignment to the variables that maximizes the number of satisfied clauses. Chou et al. (FOCS 2020) showed that $\Omega(\sqrt{n})$ space is necessary to yield a $\sqrt{2}/2+\epsilon$ approximation of the optimum value; they also presented an algorithm that yields a $\sqrt{2}/2-\epsilon$ approximation of the optimum value using $O(\epsilon^{-2}\log n)$ space.

In this paper, we focus not only on approximating the optimum value, but also on obtaining the corresponding Boolean assignment using sublinear $o(mn)$ space.  We present randomized single-pass algorithms that w.h.p. \footnote{W.h.p. denotes ``with high probability''. Here, we consider $1-1/\poly(n)$ or $1-1/\poly(m)$  as high probability.} yield:
\begin{itemize}
\item A $1-\epsilon$ approximation using $\tilde{O}(n/\epsilon^3)$ space and exponential post-processing time  
\item A $3/4-\epsilon$ approximation using $\tilde{O}(n/\epsilon)$ space and polynomial post-processing time. \end{itemize}
Our ideas also extend to dynamic streams. On the other hand, we show that the streaming $\ksat$ problem that asks to decide whether one can satisfy all size-$k$ input clauses must use $\Omega(n^k)$ space. 

We also consider the related minimum satisfiability problem ($\minsat$), introduced by Kohli et al. (SIAM J. Discrete Math. 1994), that asks to find an assignment that minimizes the number of satisfied clauses. For this problem, we give a $\tO(n^2/\epsilon^2)$ space algorithm, which is sublinear when $m = \omega(n)$, that yields an $\alpha+\epsilon$ approximation where $\alpha$ is the approximation guarantee of the offline algorithm. If each variable appears in at most $f$ clauses, we show that a $2\sqrt{fn}$ approximation using $\tO(n)$ space is possible.

Finally, for the \maxandsat problem where clauses are conjunctions of literals, we show that any single-pass  algorithm that approximates the optimal value up to a factor better than 1/2 with success probability at least $2/3$ must use $\Omega(mn)$ space.

\end{abstract}
 
 \newpage

\section{Introduction}

\paragraph{Problems overview.} The Boolean satisfiability problem (\sat) is one of the most famous problems in computer science. A satisfiability instance is a conjunction of  $m$ clauses $C_1 \land C_2 \land \ldots \land C_m$ where each clause $C_j$ is a disjunction of literals drawn from a set of $n$ Boolean variables $x_1,\ldots,x_n$ (a literal is either a variable or its negation). Deciding whether such an expression is satisfiable is \textup{NP-Complete} \cite{Cook71,Trakhtenbrot84}. When each clause has size exactly $k$, this is known as the $\ksat$ problem.

In the optimization version, one aims to find an assignment to the variables that maximizes the number of satisfied clauses. This is known as the maximum satisfiability problem (\maxsat). This problem is still \textup{NP-Hard} even when each clause has at most two literals \cite{GareyJS76}.  However, \maxsat can be approximated up to a factor $3/4$ using linear programming (LP) \cite{GW94}, network flow \cite{Y94}, or a careful greedy approach \cite{PSWZ17}. In polynomial time, one can also obtain an approximation slightly better than 3/4 using semidefinite programming (SDP) \cite{GW95,ABZ05}. Hastad showed the inapproximability result that unless $\textup{P}=\textup{NP}$, there is no polynomial-time algorithm that yields an approximation better than 21/22 to \maxtwosat \cite{Hastad01}.

A related problem is the minimum satisfiability problem ($\minsat$) which was introduced by Kohli et al. \cite{KKM94}. In this problem, the goal is to minimize the number of satisfied clauses. They showed that this problem is \textup{NP-Hard} and gave a simple randomized 2-approximation. Marathe and Ravi \cite{MaratheR96} showed that \minsat is equivalent to the minimum vertex cover problem and therefore an approximation factor better than 2 in polynomial time is unlikely. Better approximations for $\minksat$ for small values of $k$ have also been developed by Avidor and Zwick \cite{AvidorZ02}, Bertsimas et al. \cite{BertsimasTV99}, and Arif et al. \cite{ArifBGK20} using linear and semidefinite programming.

In this paper, we also consider another related optimization problem \maxandsat. This problem is similar to \maxsat except that each clause is a conjunction of literals (as opposed to being a disjunction of literals in \maxsat). Trevisan studied this problem in the guise of parallel algorithms \cite{T98}. We aim to understand the space complexity of \maxsat, \minsat, \ksat, and \maxandsat in the streaming model. 

\paragraph{The data stream model.} In this setting, clauses are presented one by one in the stream in an online fashion and  the objective is to use sublinear space $o(mn)$ while obtaining a guaranteed non-trivial approximation. 

\paragraph{Motivation and past work.} Constraint satisfaction problems and their optimization counterparts have recently received notable attention in the data stream model. Some examples include vertex coloring \cite{BeraCG20,AssadiCK19}, \textup{Max-2-AND} \cite{GVV17}, \textup{Max-Boolean-CSPs} and  \textup{Max-$k$-SAT}  \cite{ChouGV20, chou2021approximability}, \textup{Min-Ones $d$-SAT} \cite{AgrawalBBBCMM0019}, and \textup{Max-2-XOR} \cite{KK19}. 

In terms of applications, \sat, \maxsat, and \minsat have been used in model-checking, software package management, design debugging, AI planning, bioinformatics, combinatorial auctions, etc.  \cite{JSS00,ABLSR10,CSMV10,AM07,HPJED17,C08,LM06,JMJI15, LiZMS11, M08, SZGN17}. Many of these applications have inputs that are too large to fit in a single machine's main memory. Furthermore, in many applications, we need to run multiple instances of \maxsat and hence saving memory could be valuable. 

Examples of large \maxsat benchmarks that  arise from real-world applications can be found at \cite{MaxSATbenchmark}. This motivates us to study this problem in the streaming setting that aims to use sublinear memory. 

\maxsat and \maxandsat were also studied by Trevisan \cite{T98} in the guise of parallel algorithms. Trevisan showed that there is a parallel algorithm that finds a $3/4-\epsilon$ approximation to  \maxsat in $O(\poly(1/\epsilon,\log m))$ time using $O(n+m)$ processors \cite{T98}. Our results here show that it suffices to use $O(n)$ processors.

The most relevant result is by Chou et al. \cite{ChouGV20}. They showed that $\Omega(\sqrt{n})$ space is required to yield a $\sqrt{2}/2+\epsilon$ approximation of the optimum value of \maxksat for $k\geq 2$; they also presented an algorithm that yields a $\sqrt{2}/2-\epsilon$ approximation of the optimum value of \maxsat using $O(\epsilon^{-2}\log n)$ space.

In many cases, we want to not only approximate the optimum value, but also output the corresponding Boolean assignment which is an objective of this work. It is worth noting that storing the assignment itself requires $\Omega(n)$ space. While our algorithms use more space, this allows us to actually output the assignment and to obtain a better approximation.

To the best of our knowledge, unlike \maxsat,  there is no prior work on $\sat,\minsat$, and $\maxandsat$ in the streaming model.

\paragraph{Main results.} Hereinafter, the memory use is measured in terms of bits. All randomized algorithms succeed w.h.p. For \maxsat, we show that it is possible to obtain a non-trivial approximation using only $\tO(n)$ space \footnote{$\tilde{O}$ hides $\polylog$ factors.}   which is roughly the space needed to store the output assignment. Throughout this paper, algorithms actually output an  assignment to the variables along with an estimate for the number of satisfied clauses. In this paper, we solely focus on algorithms that use a single pass over the stream. Furthermore, unless stated otherwise, we assume insertion-only streams.

The algorithms for \maxsat rely on two simple observations. If $m =\omega(n/\epsilon^2)$ and we sample $\Theta(n/\epsilon^2)$ clauses uniformly at random, then w.h.p. an $\alpha$ approximation on the sampled clauses corresponds to an $\alpha-\epsilon$ approximation on the original input. Moreover, if a clause is large, it can be satisfied w.h.p. as long as each literal is set to $\true$ independently with a not-too-small probability and therefore we may ignore such clauses. This second observation also allows us to extend our result to insertion-deletion (dynamic) streams.

Based on the above observations, we proceed by simply ignoring large clauses. Then, among the remaining (small) clauses, we sample $\Theta(n/\epsilon^2)$ clauses uniformly at random, denoted by $W$. Finally, we run some randomized $\alpha$ approximation algorithm on $W$ in post-processing in which every literal is set to $\true$ with some small probability. This will lead to an $\alpha-\epsilon$ approximation on the original set of clauses w.h.p.

\paragraph{No-duplicate assumption.}There is a subtlety regarding duplicate clauses, especially for dynamic streams.  Suppose two (or more) similar clauses (i.e., duplicates) appear in the stream, would we consider them as one clause or two separate clauses (or equivalently, one clause with weight 2)? This boils down to the choice of using an $L_0$ sampler or an $L_1$ sampler. That is whether one samples a clause uniformly at random as long as it appears in the stream or based on its frequency. However, to facilitate our discussion, hereinafter, we assume that there is no duplicate in the stream. 

Our first main results are algorithms for $\maxsat$ that use space linear in terms of $n$. Note that the space to store the output assignment itself is $\Omega(n)$.

\begin{theorem}\label{thm:max-sat} 
We have the following randomized streaming algorithms for $\maxsat$.
\begin{itemize}
\item  A $3/4-\epsilon$ and a $1-\epsilon$ approximations for insertion-only streams while using $\tO(n/\epsilon)$ and $\tO(n/\epsilon^3)$ space respectively. These algorithms have $O(1)$ update time. \item A $3/4-\epsilon$ and a $1-\epsilon$ approximations for  dynamic streams while using $\tO(n/\epsilon)$ and $\tO(n/\epsilon^4)$ space respectively. The update time can be made $\tO(1)$ with an additional $\epsilon^{-1}\log n$  factor in the space use. 
\end{itemize}
\end{theorem}

The decision problem \sat is however much harder in terms of streaming space complexity. Specifically, we show that $o(m)$ space is generally not possible. Our lower bound holds even for the decision \ksat problem where each clause has exactly $k$ literals and $m= \Theta((n/k)^k)$.

\begin{theorem}\label{thm:lb-ksat}
Suppose $k \leq  n/e$. Any single-pass streaming algorithm that solves $\ksat$ with success probability at least 3/4 requires $\Omega(m)$ space where $m = \Theta((n/k)^k)$ w.h.p. 
\end{theorem}
This lower bound for \ksat is tight up to polylogarithmic factors in the following sense. If $k > \lceil \log_2 m \rceil$, we know that it is possible to satisfy all the input clauses via a probabilistic argument. In particular, we can independently assign each variable to $\true$ or $\false$ equiprobably then the probability that a clause is not satisfied is smaller than $1/m$. Hence, a union bound over $m$ clauses implies that the probability that we satisfy all the clauses is positive. On the other hand, if $k < \log m$, storing the entire stream requires $\tO(m)$ space. Hence, we have an $\tO(m)$-space algorithm. 

For \minsat, we observe that one can obtain a $1+\epsilon$ approximation in $\tO(n^2/\epsilon^2)$ space using a combination of $F_0$ sketch and brute-forte search. However, it is not clear how to run  polynomial time algorithms based on the sketch (see Section \ref{sec:min-sat} for further discussion). We provide another approach that sidesteps this entirely.

\begin{theorem}\label{thm:min-sat}
Suppose there is an $\alpha$ offline approximation algorithm for $\minsat$ that runs in $T$ time. Then, we have a randomized single-pass, $\tO(n^2/\epsilon^2)$-space streaming algorithm for \minsat  that yields an $\alpha + \epsilon$ approximation and uses $T$ post-processing time.
\end{theorem}

\paragraph{Other results.} For \minsat, if each variable appears in at most $f$ clauses, then it is possible to yield a $2\sqrt{nf}$ approximation to \minsat in $\tO(n)$ space. On the lower bound side, we show that any streaming algorithm for \minsat that decides if $\opt = 0$ must use $\Omega(n)$ space.

Finally, we present a space lower bound on streaming algorithms that approximate the optimum value  of \maxandsat up to a factor $1/2+\epsilon$. In particular, we show that  such algorithms must use $\Omega(mn)$ space.

 \paragraph{Notation and preliminaries.} We occasionally use set notation for clauses. For example, we write $x_i \in C_j$ (or $\compl{x_i} \in C_j$) if the {\em literal} $x_i$  (or $\compl{x_i}$ respectively) is in clause $C_j$. Furthermore, $C_j \setminus x_i$ denotes the clause $C_j$ with the literal $x_i$ removed. We often write $\opt(P)$ to denote the optimal value of the problem $P$, and when the context is clear, we drop $P$ and just write $\opt$. We write $\poly(x)$ to denote $x^c$  for an arbitrarily large constant $c$; in this paper, constant $c$ is absorbed in the big-$O$. Throughout this paper, we use $K$ to denote a universally large enough constant. We assume that our algorithms is provided with $m$ or an upper bound for $m$ in advance; this is necessary since the space often depends on $\log m$.
 
\paragraph{Chernoff bound.} We use the following version of Chernoff bound. If $X = \sum_{i=1}^n X_i$ where $X_i$ are negatively correlated binary random variables and $\prob{X_i=1} = p$, then for any $\eta > 0$, $\prob{\left| X  - p n \right| \geq \eta p n } \leq 2 \cdot \exp{{-\frac{\eta^2 }{2+\eta} \cdot pn }}$.

\paragraph{Organization.} We provide our main algorithms for \maxsat and \minsat in Section \ref{sec:maxsat-algorithms} and Section \ref{sec:min-sat} respectively. The space lower bound results are presented in Section \ref{sec:space-lb}; some discussions and omitted proofs are deferred to the Appendix. 

\section{Streaming Algorithms for \maxsat}\label{sec:maxsat-algorithms}
Without loss of generality, we may assume that there is no {\em trivially true} clause, i.e., clauses that contain both a literal and its negation and there are no duplicate literals in a clause.  We show that if we sample $\Theta(n/\epsilon^2)$ clauses uniformly at random and run a constant approximation algorithm for \maxsat on the sampled clauses, we obtain roughly the same approximation. Note that if $m \leq K n/\epsilon^2$ we might skip the sampling part.

\begin{lemma}\label{lem:sampling1}
For \maxsat, an $\alpha$ approximation on $K n/\epsilon^2 $ clauses sampled uniformly at random corresponds to an $\alpha-\epsilon$ approximation on the original input clauses with probability at least $1-e^{-n}$.
\end{lemma}
\begin{proof}
We recall the folklore fact that $\opt \geq m/2$. Consider an arbitrary assignment. If it satisfies fewer than $m/2$ clauses, we can invert the assignment to satisfy at least $m/2$ clauses.

Suppose that an assignment $A$ satisfies $m_A$ clauses. Let the number of sampled clauses that $A$ satisfies be $m_A'$ and let $p = Kn/(\epsilon^2 m)$. For convenience, let $y = m_A$ and $y' = m_A'$. We observe that $\expec{y'} =p y$. 
 
Suppose the assignment $A$ satisfies clauses $C_{\sigma_1},\ldots,C_{\sigma_y}$. We define the indicator variable $X_i = [\text{$C_{\sigma_i}$ is sampled}]$ and so $y' = \sum_{i=1}^y X_i$. Let $\eta = \epsilon \opt/y$.  Since we sample without replacement, $\{ X_i\}$ are  negatively correlated. Appealing to Chernoff bound, we have
\begin{align*}
\prob{|y' - p y| \geq \epsilon p \opt} &  \leq 2 \cdot \exp{-\frac{\eta^2}{2+\eta} py} \\ 
& \leq 2 \cdot \exp{-\frac{\epsilon^2 \opt^2/y^2}{2+\epsilon \opt/y}py} \\
& = 2 \cdot \exp{-\frac{\epsilon^2 \opt^2}{2y+\epsilon \opt}p}  \leq 2 \cdot \exp{-\epsilon^2\opt p/3}.
\end{align*}
The last inequality follows because $3\opt \geq 2 y + \epsilon \opt$. Therefore, 
\begin{align*}
\prob{m_A' = p  m_A \pm \epsilon p \opt } & \geq 1-  2 \cdot\exp{-\frac{\epsilon^2 p \opt}{3}}  =  1- 2 \cdot \exp{-\frac{\epsilon^2 \frac{K n}{m \epsilon^2} \opt}{3}} \nonumber \\
& \geq 1-2 \cdot \exp{\frac{-Kn}{6}}   \geq 1-\exp{-100n}.
\end{align*}
The second inequality follows from the fact that $\opt \geq m/2$. A union bound over $2^n$ distinct assignments implies that with probability at least $1-e^{-n}$, we have $m_A' = p  m_A \pm \epsilon p \opt$ for all assignments $A$.

Suppose an assignment $\tilde{A}$ is an $\alpha$ approximation to \maxsat on the sampled clauses. Let $A^\star$ be an optimal assignment on the original input clauses. From the above,  with probability at least $1-e^{-n}$, we have $pm_{\tilde{A}} + \epsilon p \opt \geq m_{\tilde{A}}' \geq \alpha m_{A^\star}' \geq \alpha \opt p(1-\epsilon) $.
Hence, $m_{\tilde{A}} \geq  \alpha \opt (1-\epsilon)  - \epsilon  \opt  \geq (\alpha - 2 \epsilon) \opt$. Reparameterizing $\epsilon \leftarrow \epsilon/2$ completes the proof.
\end{proof}

Note that storing a clause may require $\Omega(n)$ space and hence the space use can still be $\Omega(n^2/\epsilon^2)$ after sampling. We then observe that large clauses are probabilistically easy to satisfy. We define $\beta$-large clauses as clauses that have at least $\beta$ literals.

\begin{lemma} \label{lem:large-clauses}
If each literal is set to $\true$ independently with probability at least $\gamma$, then the assignment satisfies all $(K\log m)/\gamma$-large clauses with probability at least $1-1/\poly(m)$.
\end{lemma}
\begin{proof}
We can safely discard clauses that contains both  $x$ and $\overline{x}$ since they are trivially true. Thus, all literals in each clauses are set to $\true$ independently with probability at least $\gamma$. The probability that a $(K  \log m)/\gamma$-large clause is not satisfied is at most 
\[
(1-\gamma)^{(K \log m)/\gamma} \leq e^{-K \log m} \leq \frac{1}{\poly(m)}
\] 
which  implies that all such clauses are satisfied with probability at least $1-1/\poly(m)$ by appealing to a union bound over at most $m$ large clauses.
\end{proof}

From the above observations, we state a simple meta algorithm (Algorithm \ref{alg:meta1}), that can easily be implemented in several sublinear settings. We will then present two possible post-processing algorithms and the corresponding $\gamma$ values. 

~

\begin{algorithm}
\setstretch{1.35}
\DontPrintSemicolon
\caption{A  meta algorithm for sublinear \maxsat}\label{alg:meta1}
Ignore all $\beta$-large clauses where $\beta =  \frac{K \log m}{\gamma}$. Among the remaining clauses, sample and store $ K n/\epsilon^2$ clauses uniformly at random. Call this set $W$. \\ 
{\bf Post-processing:} Run an $\alpha$ approximation to \maxsat on the collected clauses $W$ where each literal is set to $\true$ independently with probability at least $\gamma$.\\
\end{algorithm}

Let $L$ and $S$ be the set of $\beta$-large and small clauses respectively. Furthermore, let $\opt_L$ and $\opt_S$ be the number of satisfied clauses in $L$ and $S$ respectively in the optimal assignment. 

\paragraph{Post-processing algorithm 1: An exponential-time $1-\epsilon$ approximation.} 

Here, we set $\gamma = \epsilon$.  Suppose in post-processing, we run the exact algorithm to find an optimal assignment $A^\star$ on the set of collected clauses $W$.  Ideally, we would like to apply Lemma \ref{lem:sampling1}  to argue that we yield a $1-\epsilon$ approximation on the small clauses in $S$ w.h.p. and then apply Lemma  \ref{lem:large-clauses} to argue that we satisfy all the large clauses $L$ w.h.p. However, the second claim requires that each literal is set to $\true$ with probability at least $\epsilon$ whereas the exact algorithm is deterministic. Our trick is to randomly perturb the  assignment  given by the exact algorithm.  

If the exact algorithm sets $x_i = q \in \{\true,\false\}$, we set  $x_i = q$ with probability $1-\epsilon$ (and $x_i = \overline{q}$ with probability $\epsilon$). We will show that this yields a $1-\epsilon$ approximation in expectation which can then be repeated $O(\frac{\log m}{\epsilon})$ times to obtain the ``w.h.p.'' guarantee. 
~

\begin{algorithm}
\setstretch{1.35}

\DontPrintSemicolon
\caption{A $1-2\epsilon$ approximation post-processing}\label{alg:processing1}
Obtain an optimal assignment $A^\star$ on $W$. Let $Q = \lceil K \epsilon^{ -1}\log m) \rceil$. \\
For each trial $t=1,2,\ldots,Q$, if {$x_i = q$ in $A^\star$}, then {set $x_i = q$ with probability $1-\epsilon$ and $x_i =  \compl{q}$ with probability $\epsilon$ in assignment $A_t$.} \\
Return the assignment $A_t$ that satisfies the most number of clauses in $W$. 
\end{algorithm}

Before analyzing the above post-processing algorithm, we observe that if we obtain an $\alpha \geq 1/2$ approximation in expectation, we can repeat the corresponding algorithm  $O(\epsilon^{-1}\log m)$ times and choose the best solution to obtain an $\alpha-\epsilon$ approximation w.h.p.

\begin{lemma}\label{lem:whp}
An $\alpha \geq 1/2$ approximation to \maxsat in expectation can be repeated $O(\epsilon^{-1}\log m)$ times to yield an $\alpha-\epsilon$ approximation w.h.p.
\end{lemma}

We show that the post-processing Algorithm \ref{alg:processing1} yields a $1-\epsilon$ approximation. 
\begin{lemma}\label{lem:post-processing-1}

Algorithm \ref{alg:processing1} yields a $1-3\epsilon$ approximation on the original input w.h.p. 
\end{lemma}
\begin{proof}
Clearly, in each trial $t=1,2,\ldots,Q$, each literal is set to $\true$ with probability at least $\epsilon$. According to Lemma \ref{lem:large-clauses}, each assignment $A_t$ satisfies all the large clauses in $L$ with probability at least $1-1/\poly(m)$. Taking a union bound over $Q < m$ trials, we conclude that all assignments $A_t$ satisfy all the clauses in $L$ with probability at least $1-1/\poly(m)$. 

Next, let $B$ be the set of clauses  in $W$ satisfied by the exact assignment $A^\star$. Consider any trial $t$. The expected number of satisfied clauses in $B$ after we randomly perturb the exact assignment is
\[
\sum_{C \in B} \prob{\text{$C$ is satisfied}} \geq \sum_{C \in B} (1-\epsilon) = (1-\epsilon)|B|.
\]
The first inequality follows from the observation that at least one literal in $C$ must be $\true$ in the exact assignment and it remains $\true$ with probability at least $1-\epsilon$. The assignment returned by post-processing Algorithm \ref{alg:processing1} yields a $1-2\epsilon$ approximation on $W$ with probability at least $1-1/\poly(m)$ by Lemma \ref{lem:whp}. This, in turn, implies that we obtain a $1-3\epsilon$ approximation on $S$ with probability at least $1-1/\poly(m)-e^{-n} \geq 1-1/\poly(m)-1/\poly(n)$ by Lemma \ref{lem:sampling1}. 

Therefore, we satisfy at least  $(1-3\epsilon)\opt_S + \opt_L \geq (1-3\epsilon)\opt$ clauses with probability at least $ 1-1/\poly(m)-1/\poly(n)$.
\end{proof}
\paragraph{Post-processing algorithm 2: A polynomial-time $3/4-\epsilon$ approximation.} We can set $\beta = K \log m$ if we settle for a $3/4-\epsilon$ approximation. This saves a factor $1/\epsilon$ in the memory use. Consider the standard linear programming (LP)  formulation for \maxsat given by Goemans and Williamson  \cite{GW94}. 
\begin{align*}
(LP) ~~    \text{maximize} & ~~ \sum_{j=1}^m z_j \\
    \text{subject to} & ~~ \sum_{i \in P_j} y_i + \sum_{i \in N_j} (1-y_i) \geq z_j && \text{ for all $1 \leq j \leq m $}\\
    & 0 \leq y_i,z_j \leq 1 && \text{ for all $1\leq i \leq n, 1\leq j \leq m$~.}
\end{align*}

The integer linear program where $y_i,z_j \in \{0,1\}$ corresponds exactly to the \maxsat problem. In particular, if $x_i \in C_j$ then $i \in P_j$ and if $\overline{x_i} \in C_j$ then $i \in N_j$. We associate $y_i=1$ with $x_i$ being set to $\true$ and $y_i=0$ if it is set to $\false$. Similarly, we associate $z_i = 1$  with clause $C_j$ being satisfied and 0 otherwise. Let $\opt(LP)$ denote the optimum of the above LP.

\begin{lemma} [\cite{GW94}, Theorem 5.3]\label{lem:rounding1} 
The optimal fractional solution of  the LP for \maxsat can be rounded to yield a 3/4 approximation in expectation by independently setting each variable to $\true$ with probability $1/4+y_i^\star/2$ where $y_i^\star$ is the value of $y_i$ in $\opt(LP)$. 
\end{lemma}

\begin{algorithm}
\setstretch{1.35}

\DontPrintSemicolon
\caption{A $3/4-2\epsilon$ approximation post-processing }\label{alg:processing2}
Obtain an optimal solution $z^\star,y^\star$ for the linear program of \maxsat on  $W$. \\
Let $Q = \lceil K \epsilon^{-1}\log m \rceil$. For each trial $t=1,2,\ldots,Q$, set $x_i = \true$ with probability $1/4+y_i^\star/2$ in assignment $A_t$. \\
Return the assignment $A_t$ that satisfies the most number of clauses in $W$. 
\end{algorithm}

The following lemma and its proof are analogous to Lemma \ref{lem:post-processing-1}. 

\begin{lemma}\label{lem:post-processing-2}
Algorithm \ref{alg:processing2} yields a $3/4-3\epsilon$ approximation on the original input w.h.p.
\end{lemma}

\paragraph{Handling deletions.} To support deletions, we use the $L_0$ sampler in \cite{JayaramW18, JST11}. Suppose each stream token is either an entry deletion or insertion on a vector $v$ of size $N$ then the $L_0$ sampler returns a non-zero coordinate $j$ uniformly at random. The sampler succeeds with probability at least $1-\delta$ and uses $O(\log^2 N \cdot \log (1/\delta))$ space. We can use a vector of size $N = 2^{2n}$ to encode the characteristic vector of the set of the clauses we have at the end of the stream. However, this results in an additional factor $n^2$ in the space use. Fortunately, since we only consider small clauses of size at most $\beta$, we can use a vector of size 
\[
    N = {2n \choose 1} + {2n \choose 2} + \ldots + {2n \choose \beta} \leq \sum_{i=1}^{\beta} (2n)^i \leq  O((2n)^{\beta}).
\]
Thus, to sample a small clause with probability at least $1-1/\poly(n)$, we need $O(\beta^2 \log^3 n )$ space. To sample $K n/\epsilon^2$ small clauses, the space is $\tO(n \beta^2/ \epsilon^2 )$. Finally, to sample clauses without replacement, one may use the approach described by McGregor et al. \cite{MTVV15} that avoids an increase in space.

A naive implementation to sample $s$ clauses is to run $s$ different $L_0$ samplers in parallel which results in an update time of $\tO(s)$. However, \cite{MTVV15} showed that it is possible to improve the update time to $\tO(1)$ with an additional factor  $\log N = O(\beta \log n)$ in the space use.

\paragraph{Putting it all together.} We finalize the proof of the first main result by outlining the implementation of the algorithms above in the streaming model. We further present an improvement that saves another $1/\epsilon$ factor under the no-duplicate assumption in the $3/4-\epsilon$ approximation.

\begin{proof}[Proof of Theorem \ref{thm:max-sat} (1)]
We ignore $\beta$-large clauses during the stream. Among the remaining clauses, we can sample and store $\Theta(n/\epsilon^2 )$  small clauses in the stream uniformly at random as described. 

For insertion-only streams, we may use  Reservoir sampling \cite{V85} to sample clauses. Since storing each small clause requires $\tO(\beta)$ space, the total space is $\tO(n \beta /\epsilon^2)$. We can run post-processing Algorithm \ref{alg:processing1}  where $\beta = K \epsilon^{-1} \log m$ and yield a $1-\epsilon$ approximation; alternatively, we set $\beta = (K \log m)$ and run the polynomial post-processing Algorithm \ref{alg:processing2} to yield a $3/4-\epsilon$ approximation.

We can further improve the dependence on $\epsilon$ for the $3/4-\epsilon$ approximation as follows. If the number of 1-literal clauses is at most $\epsilon m$, we can safely ignore these clauses. This is because the number of 1-literal clauses is then at most $2\epsilon \opt$. If we randomly set each variable to $\true$ with probability $1/2$, we will yield a $3/4$ approximation in expectation on the remaining clauses since each of these clause is satisfied with probability at least $1-1/2^2 = 3/4$. By Lemma \ref{lem:whp}, we can run $O(\epsilon^{-1}\log m)$ copies in parallel and return the best assignment to yield a $3/4-\epsilon$ approximation w.h.p. Therefore, the overall approximation is $3/4-3\epsilon$ and the space is $O(n/\epsilon \cdot \log m)$. 

If the number of 1-literal clauses is more than $\epsilon m$, then we know that $m \leq 2n/\epsilon$ since the number of 1-literal clauses is at most $2n$ if there are no duplicate clauses.  Since $m \leq 2n/\epsilon$, the sampling step can be ignored. It then suffices to store only clauses of size at most  $K \log m$ and run post-processing Algorithm \ref{alg:processing2}. The space use is $O(m \log m)=\tO(n/\epsilon)$.
\end{proof}

\begin{proof}[Proof of Theorem  \ref{thm:max-sat} (2)]
For insertion-deletion streams without duplicates, we sample clauses using the $L_0$ sampler as described earlier. We can run post-processing Algorithm \ref{alg:processing1}  where $\beta = K \epsilon^{-1}\log m$ and obtain a $1-\epsilon$ approximation while using $\tO(n/\epsilon^4)$ space. Alternatively, we set $\beta = K \log m$ and run the polynomial post-processing Algorithm \ref{alg:processing2} to yield a $3/4-\epsilon$ approximation. This leads to a $\tO(n/\epsilon)$-space algorithm that yields a $3/4-\epsilon$ approximation.
\end{proof}

We remark that it is possible to have an approximation slightly better than $3/4-\epsilon$ in polynomial post-processing time using  semidefinite programming (SDP) instead of linear programming  \cite{GW95}. In the SDP-based algorithm that obtains a $0.7584-\epsilon$ approximation, we also have the property that each literal is true with probability at least some constant. The analysis is analogous to what we outlined above. 

\section{Streaming Algorithms for \minsat}\label{sec:min-sat}
In this section, we provide several algorithms for \minsat in the streaming setting. One can show that $F_0$-sketch immediately gives us a $1+\epsilon$ approximation. However, the drawback of this approach is its inability to adapt polynomial time algorithms that yields approximations worse than $1+\epsilon$. See Appendix \ref{app:F0-minsat} for a detailed discussion.

\paragraph{The subsampling framework.} We first present a framework that allows us to assume that the optimum value is at most $O(n/\epsilon^2)$.
\begin{lemma}\label{lem:sampling-minsat}
Suppose $\frac{K n}{\epsilon^2} \leq \opt \leq z \leq 2\opt $ and $0 \leq \epsilon < 1/4$. An $\alpha$ approximation to \minsat on clauses sampled independently with probability $p = \frac{K n}{\epsilon^2 z}$ corresponds to an $\alpha+\epsilon$ approximation on the original input w.h.p. Furthermore, the optimum of the sampled clauses $\opt' = O(n/\epsilon^2)$.
\end{lemma}

Based on the above lemma, we can run any $\alpha$ approximation algorithm on sampled clauses with sampling probability $p = \frac{K n}{\epsilon^2 z}$ to yield an almost as good approximation. Furthermore, the optimum value on these sampled clauses is at most $O(n/\epsilon^2)$. 

Since we do not know $z$, we can run different instances of our algorithm corresponding to different guesses $z=1,2,4,\ldots, 2m$. At least one of these guesses satisfies $\opt \leq z \leq 2\opt$. The algorithm instances that correspond to wrong guesses may use too much space. In that case, we simply terminate those instances. We then return the best solutions among the instances that are not terminated. Hereinafter, we may safely assume that $\opt = O(n/\epsilon^2)$.

\begin{proof}[Proof of Theorem \ref{thm:min-sat}]
We first show that an $\tO(n^2/\epsilon^2)$-space algorithm exists. Let $\C(\ell) = \{ C_j : \ell \in C_j \}$ be the set of clauses that contains the literal $\ell$. It must be the case that either all clauses $\C(x_i)$ are satisfied or all clauses in $\C({ \compl{x_i}})$ are satisfied. For time being, assume we have an upper bound $u$ of $\opt$.

Originally, all variables are marked as unsettled. If at any point during the stream, $|\C(x_i)| > u$, we can safely set $x_i \leftarrow \false$; otherwise, $\opt > u$ which is a contradiction. We then say the variable $x_i$ is settled. The case in which $|\C(\compl{x_i})| >  u$ is similar. 

When a clause $C_j$ arrives, if $C_j$ contains a settled literal $\ell$ that was set to $\true$, we may simply ignore $C_j$ since it must be satisfied by any optimal solution. On the other hand, if the fate of $C_j$ has not been decided yet, we store $C_j \setminus Z$ where $Z$ is the set of settled literals in $C_j$ that were set to $\false$ at this point. For example, if $C_j = (x_1 \lor x_2 \lor \compl{x_3})$ and $x_1 \leftarrow \false$ at this point while $x_2$ and $x_3$ are unsettled, then we simply store $C_j = (x_2 \lor \compl{x_3})$. 

Since we store the input induced by unsettled variables and each unsettled variable appears in at most $2u$ clauses (positively or negatively), the space use is $\tO(nu)$. By using the aforementioned subsampling framework, we may assume that $u = O(n/\epsilon^2)$. Thus, the total space is $\tO(n^2/\epsilon^2)$.
\end{proof} 

If each variable appears in at most $f$ clauses, we have a slightly non-trivial approximation that uses less space.
\begin{theorem}\label{thm:min-sat2}
Suppose each variable appears in at most $f$ clauses. There exists a single-pass, $\tO(n)$-space algorithm that yields a $2\sqrt{fn}$ approximation to \minsat. 
\end{theorem}
\begin{proof}
It is easy to check if  $\opt = 0$ in $\tO(n)$ space. For each variable $x_i$, we keep track of whether it always appears positively (or negatively); if so, it is safe to set  $x_i \leftarrow \false$ (or $x_i \leftarrow \true$ respectively).  If that is the case for all variables, then $\opt = 0$. We run this in parallel with our algorithm that handles the case $\opt > 0$.

Now, assume $\opt \geq 1$. Since each variable involves in at most $f$ clauses. There are fewer than $\sqrt{fn}$ clauses with size larger than $\sqrt{fn}$. We ignore these large clauses from the input stream. Let $S_{x_i} = \{ C_j : |C_j| \leq \sqrt{fn} \text{ and } x_i \in C_j \}$. For each variable $x_i$, if $|S_{x_i}| \leq |S_{\compl{x_i}}|$, then we set $x_i \leftarrow \true$; otherwise, we set $x_i \leftarrow \false$. Let $\ell_i = x_i $ if $x_i$ was set to $\true$ and $\ell_i = \compl{x_i} $ if $x_i$ was set to $\false$ by our algorithm.  Furthermore, let  $\ell_i^\star = x_i $ if $x_i$ was set to $\true$ and $\ell_i^\star = \compl{x_i} $ if $x_i$ was set to $\false$ in the optimal assignment. Let $\opt_S$ be the number of small clauses that are satisfied by the optimal assignment and let $\opt$ be the number of clauses that are satisfied by the optimal assignment.

Note that $\sum_{i=1}^n |S_{\ell_i^\star}|  \leq \sqrt{fn} \opt_S$ since each small clauses belong to at most $\sqrt{fn}$ different $S_{\ell_i^\star}$. The number of clauses that our algorithm satisfies at most 
\[
\sqrt{fn} + \sum_{i=1}^n |S_{\ell_i}| \leq \sqrt{fn} + \sum_{i=1}^n |S_{\ell_i^\star}| \leq \sqrt{fn}\opt + \sqrt{fn} \opt_S
 \leq 2\sqrt{fn} \opt. \qedhere 
 \]
\end{proof}

\section{Space Lower Bounds} \label{sec:space-lb}

 \paragraph{Lower bound for $\ksat$.}  Let us first prove that any streaming algorithm that decides whether a $\ksat$ formula, where $k < \log m$,  is satisfiable requires $\Omega(\min\{m, n^k \})$ space. We consider the one-way communication problem $\ksat$ defined as follows. In this communication problem, Alice and Bob each has a set of $\ksat$ clauses. Bob wants to decide if there is a Boolean assignment to the variables that satisfies all the clauses from both sets. The protocol can be randomized, but the requirement is that the success probability is at least some large enough constant. Note that a  streaming algorithm that solves $\ksat$ yields a protocol for the $\ksat$ communication problem since Alice can send the memory of the algorithm to Bob.

We proceed with a simple claim. If $k=2$, then this claim says that there is no assignment that satisfies all of the following clauses $(x \lor y), (\compl{x} \lor y),(x \lor \compl{y}),(\compl{x} \lor \compl{y})$. The claim generalizes this fact for all $k \in \mathbb{N}$. 
 
 \begin{claim} \label{clm:1}
Consider $k$ Boolean variables $x_1,\ldots,x_k$. No Boolean assignment simultaneously satisfies all clauses in the set $S_k = \{ (\bigvee_{\ell \in S}  \compl{x_\ell} )\bigvee(\bigvee_{\ell' \notin S}  {x_{\ell'}} ) : S \subseteq [k] \}$ for all $k \in \mathbb{N}$.
 \end{claim}
 \begin{proof}
 The proof is a simple induction on $k$. The base case $k=1$ is trivial since $S_1$ consists of two clauses $x_1$ and $\compl{x_1}$. Now, consider $k > 1$. Suppose that there is an assignment that satisfies all clauses in $S_k$. Without loss of generality, suppose $x_k = \true$ in that assignment (the case $x_k =\false$ is analogous). Consider the set of clauses $A = \{ \compl{x_k} \vee \phi : \phi \in S_{k-1} \} \subset S_k$. Since $\compl{x_k} =\false$, it must be the case that all clauses in $S_{k-1}$ are satisfied which contradicts the induction hypothesis. 
 \end{proof}
 
Without loss of generality, we assume $n/k$ is an integer.  Now, we consider the  one-way communication problem $\ind$. In this problem, Alice has a bit-string $A$ of length $t$ where each bit is independently set to 0 or 1 uniformly at random. Bob has a random index $i$ that is unknown to Alice. For Bob to output $A_i$ correctly with probability 2/3, Alice needs to send a message of size $\Omega(t)$ bits \cite{KremerNR99}.   We will show that a protocol for for $\ksat$ will yield a protocol for $\ind$.

\begin{proof}[Proof of Theorem \ref{thm:lb-ksat}]
Without loss of generality, assume $n/k$ is an integer. Consider the case where Alice has a bit-string $A$ of length ${(n/k)^k}$. For convenience, we index $A$ as $[n/k]\times \ldots \times [n/k] = [n/k]^k$. We consider $n$ Boolean variables $\{x_{a,b} \}_{a \in [k], b \in [n/k]}$. 

For each $j \in [n/k]^k$ where $A_j = 1$, Alice generates the clause $(x_{1,j_1} \lor x_{2,j_2} \lor \ldots \lor x_{k,j_k})$.  Let $S = \{ (1,i_1),(2,i_2)\ldots,(k,i_k) \}$. Bob, with the index $i \in [n/k]^k$, generates the clauses in 
\[\left\{ \left(\bigvee_{l \in Z}  \compl{x_l} \right) \left(\bigvee_{l' \in S \setminus Z}  {x_{\ell'}} \right) : Z \subseteq S \right\} \setminus \{(x_{1,i_1} \lor x_{2,i_2} \lor \ldots \lor x_{k,i_k}) \}.\]

If $A_i = 0$, then all the generated clauses can be satisfied by setting the variables $x_{l}  = \false$ for all $l \in S$ and all remaining variables to $\true$. All the clauses that Bob generated are satisfied since each contains at least one literal among $\compl{x_{1,i_1}},\compl{x_{2,i_2}} \ldots, \compl{x_{k,i_k}}$. Note that the clause $(x_{1,i_1} \lor  x_{2,i_2} \lor \ldots \lor x_{k,i_k})$ does not appear in the stream since $A_i = 0$. For any $j \in [n/k]^k$, where $(j_1,\ldots,j_k) \neq (i_1,\ldots,i_k)$, such that $A_j = 1$, there must be some $j_z \notin \{ i_1,i_2,\ldots,i_k \}$. Hence, $x_{z, j_z}=\true$ and therefore the clause $(x_{1,j_1} \lor \ldots \lor x_{k,j_k})$ is satisfied.

If $A_i = 1$, then by Claim \ref{clm:1}, we cannot simultaneously satisfy the clauses 
\[ \left\{ \left(\bigvee_{l \in Z}  \compl{x_l} \right) \left(\bigvee_{l' \in S \setminus Z}  {x_{\ell'}} \right) : Z \subseteq S \right\} \] that were generated by Alice and Bob. Hence, any protocol for $\ksat$ yields a  protocol for $\ind$. Hence, such a protocol requires $\Omega(n^k)$ bits of communication.

Finally, we check the parameters' range where the lower bound operates.  We first show that in the above construction, the number of clauses $m$ concentrates around $\Theta((n/k)^k)$. First, note that for fixed $n$, the function $(n/k)^k$ is increasing in the interval  $k \in [1,n/e]$. Since we assume $k \leq n/e$, we have $(n/k)^k \geq n$. Let $m_A$ be the number of clauses Alice generates. Because independently each $A_j = 1$ with probability 1/2, by Chernoff bound, we have 
\[
\prob{| m_A - 0.5(n/k)^k| > 0.01 (n/k)^k} \leq \exp{-\Omega(({n}/{k})^k)} \leq \exp{-\Omega(n)}. 
\]
Bob generates $m_B = 2^k -1 $ clauses. Thus, w.h.p., $
m = m_A + m_B = \left( 0.5(n/k)^k + 2^k - 1  \right) \pm 0.01 (n/k)^k.$
Note that $2^k \leq (n/k)^k$ since $k \leq n/e < n/2$. As a result, w.h.p.,
$0.49(n/k)^k -1  \leq m \leq 1.51(n/k)^k +1$ which implies $m=\Theta((n/k)^k)$.
\end{proof}

\paragraph{Lower bound for $\maxandsat$.} We show that a $(1/2+\epsilon)$-approximation for $\maxandsat$ requires $\Omega(mn)$ space. In particular, our lower bound uses the algorithm to distinguish the two cases $\opt =1 $ and $\opt = 2$. This lower bound is also a reduction from the $\ind$ problem. Recall that in \maxandsat, each clause is a conjunction of literals. We first need a simple observation. 

\begin{claim} \label{clm:clause-systems}
Let $T=K \log m$. There exists a set of $m$ (conjunctive) clauses $C_1,\ldots,C_m$ over the set of Boolean variables $\{z_1,\ldots,z_T\}$ such that exactly one clause can be satisfied.  
\end{claim}
\begin{proof}
We show this via a probabilistic argument. For $i = 1,2,\ldots,m$ and $j \in [T]$, let $C_i = \land_{j=1}^{T} l_{i,j}$ where each $l_{i,j}$ is independently set to $z_j$ or $\compl{z_j}$ equiprobably. 

Two different clauses $C_i$ and $C_{i'}$ can be both satisfied if and only if for all $j=1,2,\ldots, T$, it holds that the variable $z_j$ appears similarly (i.e., as $z_j$  or as $\compl{z_j}$) in both clauses. This happens with probability $(1/2)^{K \log m} = 1/\poly(m)$. Hence, by a union bound over ${m \choose 2}$ pairs of clauses, exactly one clause can be satisfy with probability at least $1-1/\poly(m)$.
\end{proof}

Alice and Bob agree on such a set of clauses $C_1,\ldots,C_m$  over the set of Boolean variables $\{z_1,\ldots,z_T\}$ (which is hard-coded in the protocol). Suppose Alice has a bit-string of length $mn$, indexed as $[m] \times [n]$. For each $k \in [m]$: if $A_{k,j}=1$, she adds the literal $x_j$ to clause $D_k$ and if $A_{k,j}=0$, she adds the literal $\compl{x_j}$ to clause $D_k$. Finally, she concatenates $C_k$ to $D_k$. More formally, for each $k \in [m]$,
\[
D_k = \left( \bigwedge_{j\in [n]: A_{k,j} = 1} x_j \right) \bigwedge \left( \bigwedge_{j \in [n]: A_{k,j} = 0} \compl{x_j} \right)  \bigwedge  \left(C_k \right).
\]

Bob, with the index $(i_1,i_2) \in [m] \times [n]$, generates the clauses $D_B = x_{i_2} \wedge C_{i_1}$. If $A_{i_1,i_2} = 1$, then both clauses $D_B$ and $D_{i_1}$ can be satisfied since the set of literals in $D_B$ is a subset of the set of literals in $D_{i_1}$. If $A_{i_1,i_2} = 0$, exactly one of $D_B$ and $D_{i_1}$ can be satisfied since $x_{i_2} \in D_B$ and $\compl{x_{i_2}} \in D_{i_1}$. Furthermore, by Claim \ref{clm:clause-systems}, we cannot simultaneously satisfy $D_B$ and any other clause $D_{k}$ where $k \neq i_1$ (since $C_{i_1}$ is part of $D_B$ and $C_k$ is part of $D_k$). Therefore, if $\opt = 2$, then $A_{i_1,i_2} = 1$ and if $\opt = 1$ then $A_{i_1,i_2} = 0$. Thus, a streaming algorithm that distinguishes the two cases yields a protocol for the communication problem. The number of variables in our construction is $n + K \log m \leq 2 n$, if we assume $n \geq K \log m$. Reparameterizing $n \leftarrow n/2$ gives us the theorem below.

\begin{theorem}
Suppose $n = \Omega(\log m)$. Any single-pass streaming algorithm that decides if $\opt = 1$ or $\opt = 2$ for $\maxandsat$ with probability 2/3 must use $\Omega(mn)$ space. 
\end{theorem}

\paragraph{Lower bound for \minsat.} For \minsat, it is easy to show that deciding if $\opt > 0$ requires $\Omega(n)$ space. The algorithm in the proof of Theorem \ref{thm:min-sat2} shows that this lower bound is tight.

\begin{theorem}\label{thm:lb-minsat}
Any single-pass streaming algorithm that decides if $\opt >0$ for $\minsat$ with probability 2/3 must use $\Omega(n)$ space. 
\end{theorem}
\begin{proof}
Consider an Index instance of length $n$. If $A_j=1$, Alice adds the literal $x_j$ to the clause $C_1$; otherwise, she adds the literal $\compl{x_j}$ to $C_1$. Alice then puts $C_1$ in the stream. Bob, with the index $i$, puts the clause $C_2 = (x_i)$ in the stream. If $A_i = 0$, then $\compl{x_i} \in C_1$ and ${x_i} \in C_2$ and therefore $\opt >0$. Otherwise, it is easy to see that we have an assignment that does not satisfy any of $C_1$ or $C_2$.
\end{proof}

\bibliographystyle{plain}
\small
\bibliography{ref}
\normalsize

\appendix

\section{An $F_0$-sketch based approach for \minsat}\label{app:F0-minsat}
We show that using $F_0$ sketch, we can obtain a $1+\epsilon$ approximation in $\tO(n^2/\epsilon^2)$ space. Given a vector $x\in {\mathbb R}^n$, $F_0(x)$ is defined as the number of elements of $x$ which are non-zero. Consider a subset $S\subseteq \{1, \ldots, n\}$, let  $x_S\in \{0,1\}^n$ be  characteristic vector of $S$ (i.e., $x_i=1$ iff $i\in S$). Note that $F_0(x_{S_1} + x_{S_2}+\ldots )$ is exactly the coverage $|S_1\cup S_2\cup \ldots |$. We will use the following  result for estimating $F_0$.

 \begin{theorem}[$F_0$ Sketch, \cite{ BJKST02}]\label{thm:F0-approximation}
Given a set $S\subseteq [n]$, there exists an $\tO(\epsilon^{-2}\log \delta^{-1})$-space algorithm that  constructs a data structure $\mathcal{M}(S)$ (called an \emph{$F_0$ sketch} of $S$). The sketch has the property that the number of distinct elements in a collection of sets $S_1, S_2, \ldots, S_t$ can be approximated up to a $1 + \epsilon$ factor with  probability at least $1-\delta$ provided the collection of $F_0$ sketches $\mathcal{M}(S_1), \mathcal{M}(S_2), \ldots, \mathcal{M}(S_t)$.
\end{theorem}

The above immediately gives us a $1+\epsilon$ approximation in $\tO(n^2/\epsilon^2)$ space. Each literal $\ell$ corresponds to a set $S_\ell$ that contains the clauses that it is in, i.e., $S_\ell = \{ C_j : \ell \in C_j \}$. Hence, the goal is to find a combination of $\ell_1, \ell_2, \ldots, \ell_n$ where $\ell_i \in \{ x_i, \compl{x_i} \}$ such that the coverage $|C_{\ell_1} \cup \ldots \cup C_{\ell_n}|$ is minimized. We can construct $\mathcal{M}(S_{x_i})$ and $\mathcal{M}(S_{\compl{x_i}})$ for each $i=1,2,\ldots,n$ with failure probability $\delta = 1/(n 2^n)$. Then, we return the smallest estimated coverage among $2^n$ such combinations based on these sketches. This is a $1+\epsilon$ approximation w.h.p. since for all $\ell_1, \ell_2, \ldots, \ell_n$ where $\ell_i \in \{ x_i, \compl{x_i} \}$, we have an estimate $ (1\pm \epsilon)| S_{\ell_1} \cup \ldots \cup S_{\ell_n}| $ with probability $1-1/n$ by a simple union bound.

This approach's drawback is its exponential post-processing time. To see why this is hard to extend to other offline algorithms, let us review the algorithm by Kohli et al. \cite{KKM94} that yields a 2-approximation in expectation. The algorithm  goes through the variables $x_1,x_2,\ldots,x_n$ one by one. At step $i$, it processes the variable $x_i$. Let $a_i$ and $b_i$ be the number of newly satisfied clauses if we assign $x_i$ to $\true$ and $\false$ respectively. We randomly set $x_i \leftarrow  \true$ with probability $\frac{b_i}{a_i+b_i}$ and set $x_i \leftarrow \false$ with probability $\frac{a_i}{a_i + b_i}$. The algorithm updates the set of satisfied clauses and move on to the next variable. A simple induction shows that this is a 2-approximation in expectation. Note that we can run this algorithm $O(\epsilon^{-1}\log n)$ times and return the best solution to obtain a $2+\epsilon$ approximation w.h.p. 

Unfortunately, $F_0$ sketch does not support set subtraction (i.e., if we have the sketches $\mathcal{M}(A)$ and $\mathcal{M}(B)$, we are not guaranteed to get a $1\pm \epsilon$ multiplicative approximation of $|A \setminus B|$) and thus it is unclear how to compute or estimate $a_i$ and $b_i$ at each step $i$. 

Besides, better approximations to $\minksat$, for small values of $k$, have also been developed \cite{AvidorZ02,BertsimasTV99,ArifBGK20} using linear and semidefinite programming. It is not clear how to combine $F_0$ sketch with these approaches either. We now show how to sidestep the need to use $F_0$ sketch entirely and run any offline algorithm of our choice.

\section{Omitted Proofs}

\begin{proof}[Proof of Lemma \ref{lem:whp}]
Let $Z$ be the number of clauses satisfied by the algorithm. First, note that since $m/4 \leq \alpha \opt $ (since $\alpha \geq 1/2$), we have $\alpha \opt/(m-\alpha \opt) \geq 1/3$. 
We have $\expec{m-Z} \leq m-\alpha \opt$. Appealing to Markov inequality,
\begin{align*}
\prob{Z \leq (1-\epsilon) \alpha \opt} & = \prob{m-Z \geq m- (1-\epsilon) \alpha \opt}  \\
& \leq \frac{m-\alpha \opt}{m- (1-\epsilon) \alpha \opt} \\
& = \frac{1}{1 +\epsilon \alpha \frac{\opt}{(m-\alpha \opt)}} \leq   \frac{1}{1+\epsilon/3}.
\end{align*}
Hence, we can repeat the algorithm $O(\log_{1+\epsilon/3} m) = O(1/\epsilon \cdot \log m)$ times and choose the best solution to obtain an  $\alpha-\epsilon$ approximation with probability at least $1-1/\poly(m)$.
\end{proof}

\begin{proof}[Proof of Lemma \ref{lem:post-processing-2}]
In each trial $t=1,2,\ldots,Q$, each literal is set to $\true$ with probability at least $1/4$. This is because $1/4 \leq 1/4+y_i^\star/2 \leq 3/4$. Appealing to Lemma \ref{lem:large-clauses} with $\gamma = 1/4$, all assignments $A_t$ satisfy all the large clauses in $L$ with probability at least $1-1/\poly(m)$. The rest of the argument is then similar to that of Lemma \ref{lem:post-processing-1}. 
\end{proof}

\begin{proof}[Proof of Lemma \ref{lem:sampling-minsat}]
Suppose that an assignment $A$ satisfies $m_A$ clauses. Let the number of sampled clauses that $A$ satisfies be $m_A'$. For convenience, let $y = m_A$ and $y' = m_A'$. We observe that $\expec{y'} =p y$. Note that $y \geq \opt \geq z/2$. Suppose the assignment $A$ satisfies clauses $C_{\sigma_1},\ldots,C_{\sigma_y}$. We define the indicator variable $X_i = [\text{$C_{\sigma_i}$ is sampled}]$ and so $y' = \sum_{i=1}^y X_i$. Appealing to Chernoff bound, 
\begin{align*}
\prob{|y' - p y| \geq \epsilon p y} &  \leq 2  \exp{-\frac{\epsilon^2}{3} py}  \leq 2  \exp{-\frac{K \epsilon^2 n y}{z \epsilon^2}} \leq \exp{-100n}.
\end{align*}
Taking a union bound over $2^n$ distinct assignments, we have $\prob{m_A' = (1 \pm \epsilon)p  m_A }  \geq 1-  \exp{-50n}$ for all assignments $A$. We immediately have that $\opt' \leq (1+\epsilon)p \opt = O(n/\epsilon^2)$.

Suppose an assignment $\tilde{A}$ is an $\alpha$ approximation to \minsat on the sampled clauses. Let $A^\star$ be an optimal assignment on the original input clauses. From the above,  with probability at least $1-e^{-50n}$, we have $(1-\epsilon)pm_{\tilde{A}}   \leq  m_{\tilde{A}}' \leq \alpha m_{A^\star}' \leq \alpha \opt p(1+\epsilon) $.
Hence, $m_{\tilde{A}} \leq  (1+3\epsilon) \alpha \opt$. Reparameterizing $\epsilon \leftarrow \epsilon/3$ completes the proof.
\end{proof}

\end{document}